\documentclass{birkjour}
\newtheorem{theorem}{Theorem}[section]
\newtheorem{corollary}[theorem]{Corollary}
\newtheorem{lemma}[theorem]{Lemma}
\newtheorem{proposition}[theorem]{Proposition}
\theoremstyle{definition}

\newtheorem{assumption}[theorem]{Assumption}
\theoremstyle{remark}

\numberwithin{equation}{section}

\usepackage{mathtools, enumerate}

\def\rr{{\mathbb{R}}}
\def\cc{{\mathbb{C}}}
\def\nn{{\mathbb{N}_0}}
\def\cp{{\mathbb{C}_+}}
\def\ee{{\mathbb{E}\,}}
\def\pp{{\mathbb{P}}}
\def\zel{{z_\ell}}
\def\p{{\bf p}} 
\def\rg{{\mathcal{R}}}
\newcommand{\im}{\operatorname{Im }}

\newcommand{\supp}{\operatorname{supp }}
\newcommand{\tr}{\operatorname{Tr }}

\newcommand{\ketbra}[1]{{|#1\rangle\langle #1|}}
\newcommand{\beq}[1]{ \begin{equation} \label{#1} }
\newcommand{\eeq}{ \end{equation} }

\begin{document}
%
%
%
%
%
%
%
%
%

\title[Hierarchical Anderson Model]{Renormalization Group Analysis of the\\ Hierarchical Anderson Model}
\author{Per von Soosten}
\address{Zentrum Mathematik, TU M\"{u}nchen\\
Boltzmannstr. 3, 85747 Garching, Germany}
\email{vonsoost@ma.tum.de}
\author{Simone Warzel}
\address{Zentrum Mathematik, TU M\"{u}nchen\\
Boltzmannstr. 3, 85747 Garching, Germany}
\email{warzel@ma.tum.de}
\subjclass{47B80, 82B44}
\keywords{Anderson localization, eigenvalue statistics, renormalization group, hierarchical Anderson model}
\date{October 31, 2016}
\begin{abstract}
We apply Feshbach-Krein-Schur renormalization techniques in the hierarchical Anderson model to establish a criterion on the single-site distribution which ensures exponential dynamical localization as well as positive inverse participation ratios and Poisson statistics of eigenvalues. Our criterion applies to all cases of exponentially decaying hierarchical hopping strengths and holds even for spectral dimension $ d > 2 $, which corresponds to the regime of transience of the underlying hierarchical random walk. This challenges recent numerical findings that the spectral dimension is significant as far as the Anderson transition is concerned.
\end{abstract}
\maketitle

\parindent = 4mm 
\section{Introduction} \label{intro}
Ever since their introduction by F.~Dyson~\cite{MR0436850,MR0295719}, hierarchical models have provided cornerstones in the landscape of analytically tractable systems in statistical mechanics. Their analysis is facilitated by a renormalization group in the spirit of Wilson and Fisher, which becomes rigorous, though non-trivial, in the hierarchical case. 
As demonstrated in many special cases, hierarchical models exhibit phase transitions quite analogous to their relatives on the finite-dimensional lattices $ \mathbb{Z}^d$ (cf.~\cite{BBS16,MR1552598,MR0295719} and references therein). In fact, one of the key features of hierarchical models is an effective dimension which is a tunable parameter.  

It is therefore quite natural to try to shed light on the Anderson transition by studying the hierarchical version of Anderson's model. This transition pertains to sharply separated regions in the energy-disorder phase diagram of a quantum particle in a random medium. Its features include localization vs. delocalization properties of the eigenvectors as well as an accompanying change in the statistical properties of the eigenvalues  (Poisson vs. random matrix statistics). While the localization side of the story has by now been understood  fairly well for any graph, and in particular for $ \mathbb{Z}^d $, proofs of the existence of a delocalized phase for $ d > 2 $ remain elusive and only pertain to special situations like tree graphs or other toy models (cf.~\cite{MR3364516} and references therein; see also~\cite{RevModPhys.80.1355}). 

The study of the hierarchical Anderson model goes back to A.~Bovier~\cite{MR1063180} who investigated its density of states - a quantity which however does not contain any information about possible phase transitions. He nevertheless conjectured the appearance of a special energy at which delocalized states persist under weak disorder for (effective spectral) dimension $ d > 4 $. The first proof of complete spectral localization for the hierarchical model in the special case of a Cauchy random potential of arbitrary strength is due to  S.~Molchanov \cite{MR1463464}. His proof ideas where later extended to more general distributions by E.~Kritchevski \cite{MR2352276,MR2413200}. Notably, neither of these works proved or even conjectured the appearance of a delocalized phase in the regime of long-range, but summable hopping strength of the hierarchical Laplacian. This belief was recently challenged by F.~Metz, L.~Leuzzi, G.~Parisi, and V.~Sacksteder~\cite{PhysRevB.88.045103}, who reported numerical evidence for the appearance of a special energy in (effective spectral) dimension $ d > 2 $, for which delocalized eigenvectors appear in the case of a weak Gaussian random potential. Superficially, this numerical result looks related to the existence of resonant delocalization at special energies in yet another toy version of the Anderson model, namely that on the complete graph~\cite{MR3383321}. One of the main aims of this paper is to argue that as long as the hierarchical hopping is summable, all states are localized in every possible specification of that term and the eigenvalues exhibit Poisson statistics. Notably, this covers the case of  (effective spectral) dimension $ d > 2 $ and hence shows that, in this sense, the hierarchical Anderson model breaks the ranks of its counterparts in statistical mechanics for which this effective dimension proved significant for phase transitions. 

The findings in this paper are in agreement with the claims of~\cite{monthusgarel}, in which multifractality of the eigenfunctions and intermediate eigenvalue statistics are found numerically for the hierarchical Anderson model with critically non-sum\-mable hopping strengths of alternating sign. In fact, the conjectured location of the hierarchical Anderson transition at non-summable hopping strength is also in agreement with similar behavior in hierarchical random matrix models~\cite{FYOSSRO11,1742-5468-2011-03-L03001} and power-law banded random matrices (PBRM)~\cite{RevModPhys.80.1355}. (See also~\cite{FyKuWe15}  for a recent rigorous analysis of the critical point in PBRMs.)

Let us now specify the details of the hierarchical Anderson model and state our main assumptions in this paper. Consider the configuration space $\nn = \{0,1,2,...\}$ endowed with the (ultra)metric
\[d(j,k) = \min  \left\{ r \geq 0 \, | \, \mbox{$j $ and $ k $ belong to a common member of $\mathcal{P}_r$} \right\},\]
where $\{\mathcal{P}_r\}$ is the nested sequence of partitions defined by
\[\nn = \{0, ..., 2^{r}-1\} \cup \{2^r, ..., 2\cdot 2^r - 1\} \cup ...\]
Thus the ball $ B_r(j) = \{ k \in \nn \, | \, d(j,k) \leq r \} $ is precisely the member of $ \mathcal{P}_r$ containing $j \in \nn$. Each partition $\mathcal{P}_r$ induces an averaging operator $E_r: \ell^2 \to \ell^2$ defined by
\[E_r \psi (j) = 2^{-r} \sum_{k \in B_r(j)} \psi(k),\]
and, taking linear combinations, we obtain a family of hierarchical Laplacians
\begin{equation}\label{eq:Lapl}
\Delta = \sum_{r \geq 1} p_r E_r 
\end{equation}
indexed by a summable sequence $ \p = \left( p_r \right)_{r \geq 1}  $. 
We will restrict our attention to the case in which
\[ |p_r| \leq  \epsilon \, 2^{-cr} \]
with $ c > 0 $ and $ \epsilon > 0 $ (not necessarily small).

The definition of the hierarchical Laplacian is essentially motivated by keeping only those eigenfunctions in the spectral representation of the finite-difference Laplacian on ${\mathbb{Z}}^d$ whose periods have length $2^r$, and then splitting each such eigenfunction into orthogonal translations of a single period. The following properties are easy to check (see, e.g.~\cite{MR2352276}):
\begin{enumerate}
\item The spectral decomposition of $ \Delta $ reads
\[\Delta = \sum_{r=1}^\infty \lambda_r P_r\]
with eigenprojections $P_r =  E_r-E_{r+1} $. The corresponding eigenvalues $ 0, \lambda_1= p_1 , \dots , \lambda_n= \sum_{r=1}^n p_r , \dots $ are infinitely degenerate and accumulate at $ \lambda_\infty :=  \sum_{r=1}^\infty p_r $. The eigenfunctions become delocalized as $\lambda \to \lambda_\infty $. 
\item For the special case $\epsilon^\prime 2^{-cr} \le p_r \le \epsilon 2^{-cr}$, the decay rate of $ \p $ is linked to the spectral dimension $ d_s $  of the hierarchical Laplacian via the formula
\[
d_s  \vcentcolon =  \lim_{\lambda \downarrow 0 } \ \frac{\ln \ \langle \delta_k , 1_{[\lambda_\infty-\lambda, \lambda_\infty]}(\Delta) \delta_k \rangle }{\ln \sqrt{\lambda}} = \frac{2}{c}.
\]
The definition of $ d_s $ is motivated by comparing the behavior of the spectral measure of the Laplacian associated with any localized vector,
\[ \delta_k(x) \vcentcolon = \begin{cases} 1 , & x=k \\ 0 , &x \neq k \end{cases}\;, \] in the vicinity of the upper spectral edge to the corresponding result for the $ d $-dimensional lattice. Notably, and analogously to $ \mathbb{Z}^d $, the random walk generated by $\Delta $ is recurrent if $ d_s \leq 2 $ and transient if $ d_s > 2 $. 
\item The hopping strength between two sites $ j,k \in \nn $ satisfies
\begin{equation}\label{eq:LaplaceME}
 \left| \langle \delta_j , \Delta \delta_k \rangle \right| \leq \sum_{r\geq d(j,k)} \frac{|p_r|}{2^r}  . 
 \end{equation}
If $ p_r  = \epsilon \, 2^{-cr} $ we 
have $ | \langle \delta_j , \Delta \delta_k \rangle  | =  \frac{\epsilon}{1-2^{-1-c}}  \, [2^{-d(j,k)}]^{1+c} $, from which we identify $ 1+c $ as the decay exponent in this long-range hopping model. From this point of view the condition $ c > 0 $ corresponds to summable decay. 
\end{enumerate}

The hierarchical Anderson model refers to the random Hamiltonian
\[H(\omega) =  \Delta + V(\omega),\]
on $ \ell^2 = \ell^2(\mathbb{N}_0) $ where $ V(\omega) \delta_k = V_k(\omega) \delta_k $ is given in terms of  independent random variables $\{V_k\}$ with a common density $\varrho \in L^\infty$. We will assume that $\varrho$ decays at least as fast as a Cauchy distribution, which means that there exists some $C < \infty$ such that
\begin{equation} \label{rhodecay}
\varrho(v) \leq \frac{C}{1+v^2}
\end{equation}
for all $ v \in \mathbb{R} $.
Thus $H$ depends essentially on two parameters: the sequence $ \p \in \ell^1 $ and the single-site density $\varrho \in L^1\cap L^\infty$. 
Inspired by \cite{PhysRevB.88.045103,monthusgarel}, our analysis of this model is based on the renormalization group transformation
\begin{equation} \label{eq:rgdefinition}
  \mathcal{R} (\p, \varrho) = \left( (p_{r+1})_{r\geq 1} , T_{p_1} \varrho \right)
\end{equation}
on the parameter space $\ell^1 \times L^1$. The operator $T_p $ in~\eqref{eq:rgdefinition} maps a probability density $\varrho$ to the probability density of the random variable
\[\left(\frac{1}{2V} + \frac{1}{2 V^\prime} \right)^{-1} + p,\]
where $V$ and $V^\prime$ are independent copies of random variables with density $\varrho$. As will be explained in detail in Sections~\ref{renormalizationsection} and~\ref{localizationsection} below, the renormalization group $ \mathcal{R} $, when implemented on the level of operators $ H $, drives towards the large disorder regime since it reduces the strength of the Laplacian relative to the random potential. 
The full analysis hinges on the following non-concentration hypothesis of the density $\varrho$ under the renormalization dynamics. 
\begin{assumption} \label{nonconcentration} Let $I \subset \rr$ and set $ \varrho_E = \varrho( \cdot + E ) $. There is  $\delta > 0$ such that
\[ \sup_{E \in I }  \|T_{ p_r} ... T_{  p_1} \varrho_E\|_\infty = \mathcal{O} (2^{(c- \delta) r}) \]
(using Landau's $ \mathcal{O} $ notation).
\end{assumption}
We will show in Appendix~\ref{densityappendix} that  Assumption~\ref{nonconcentration} is satisfied with $I = \rr$ whenever $|p_r| \le \epsilon 2^{-cr} $ for some arbitrary $ \epsilon > 0 $  and one of the following is true:
\begin{itemize}
\item spectral dimension $ d_s < 2 \; $ i.e.\ $c > 1$,
\item $V$ has a Gaussian distribution and $ d_s < 4\; $  i.e.\ $c > 1/2$,
\item $V$ has a Cauchy component and arbitrary $ d_s \in (0, \infty) \; $ i.e.\  $c > 0$.
\end{itemize}
The second point refutes the idea that the spectral dimension $ d_s = 2 $ is special also in the Gaussian case. Strictly speaking, it leaves open the possibility of a phase transition at $ d_s = 4 $ as suggested by A.~Bovier~\cite{MR1063180}. However, simple numerical simulations in the Gaussian case (and beyond) indicate that Assumption~\ref{nonconcentration} is satisfied generally (for energies $ E $ in the spectrum) provided only $c > 0$. A proof of this general fact has remained elusive. We hope that our present work will stimulate further research in this direction.

It is possible to define hierarchical Anderson models corresponding to more general hierarchical structures than the one above, see \cite{MR2352276,MR2276652}, but since the analysis of such generalizations does not require any fundamentally new ideas we restrict ourselves to the present situation here.

\section{Main Results}\label{mainresults}

Our main results will be formulated in terms of restrictions of the hierarchical Anderson Hamiltonian to finite volumes  $\Lambda \subset  \nn$. Since the bounds derived are uniform in the volume, they carry over to the infinite-volume Hamiltonian. More specifically,
we consider $H_\Lambda = 1_\Lambda H 1_\Lambda$ where $1_\Lambda$ is the indicator function of $\Lambda$. This operator acts naturally on both $\ell^2$ and $\ell^2(\Lambda)$ and, as both perspectives are useful, we shall switch back and forth without further comment in the future. Most important is $\Lambda = B_n(0) =\vcentcolon B_n$ in which case we will denote $H_{B_n}$ simply by
\begin{equation}
H_n(\omega) = \sum_{r=1}^n p_r E_r 1_{B_n} + \alpha_n  \ketbra{\varphi_n} + V(\omega)\ 1_{B_n}  \, .  
\end{equation}
Here $\alpha_n \vcentcolon = \sum_{r > n} 2^{n-r} p_r$  and $   \ketbra{\varphi_n} = 1_{B_n} E_n 1_{B_n}  $ is the orthogonal projection onto the maximally delocalized vector $ \varphi_n \in \ell^2(B_n)$. Note that the factor $ \alpha_n  \ketbra{\varphi_n} $ stems from the terms $ r \geq n+1 $ in the sum in~\eqref{eq:Lapl}. It is customary to also consider $ H_{n,n} = \sum_{r=1}^n p_r E_r 1_{B_n}  + V\ 1_{B_n}  $ (which we will occasionally do in the proof) and this amounts only to a redefinition of $ p_{n} $.

The first main result of this paper concerns localization in terms of the eigenfunction correlator (EC). For any self-adjoint operator $ H $, the EC is the total variation of its spectral measure of $\delta_j$ and $\delta_k$ restricted to $I \subset \rr$, i.e.
\[Q(j,k; I) = \sup \, |\langle \delta_k, f(H) \delta_j \rangle | ,\]
the supremum being taken over all $f \in C_0$ with $\supp f \subset I$ and $\|f\|_\infty \le 1$. In finite volume, for which the spectrum $ \sigma(H_n) $ of $ H_n $ is known to comprise  almost surely of finitely many non-degenerate eigenvalues, the eigenfunction correlator turns out to be $ Q_n(j,k;I) = \sum_{E \in \sigma(H_n)} | \psi_{E,n}(j)| |\psi_{E,n}(k)| $, where
$ \psi_{E,n} $ denotes the normalized eigenfunction of $ H_n $  corresponding to $E $.

\begin{theorem}[EC localization] \label{localizationthm} If Assumption~\ref{nonconcentration} is satisfied in a bounded interval $I \subset \rr$, then there exist $C ,\mu \in (0,\infty) $ such that
\begin{equation}\label{eq:ECloc} \sup_{n \in \mathbb{N} } \sup_{j \in \nn}   \sum_{k \in \nn} 2^{\mu d(j,k) } \ee[ Q_n(j,k; I)] \leq C | I | .
\end{equation}
\end{theorem}

The proof of this theorem can be found in Section~\ref{localizationsection}.

By the lower semicontinuity property $Q(j,k; I) \leq \liminf_{n\to \infty}   Q_n(j,k; I) $ for any open $ I \subset \mathbb{R} $, the result~\eqref{eq:ECloc} 
extends to the eigenfunction correlator for the infinite volume (cf.~\cite[Ch.~7]{MR3364516}). As a consequence, there exists some $C > 0$ such that
\[  \sum_{k : d(j,k) \geq R} \ee |\langle \delta_k, 1_I(H) e^{itH} \delta_j \rangle | ^2 \leq C \, 2^{-\mu R},\]
which shows that the quantum probability that a particle, which was started at $j \in \nn$ and subsequently filtered by energy, ever leaves $B_R(j)$ decays exponentially in $R$. 
In particular, this implies that the spectrum of $H$ is almost surely of pure-point type and that the corresponding normalized eigenfunctions decay exponentially with respect to the hierarchical metric. More precisely, there is some random amplitude $ A(\omega) \in (0,\infty) $ with $ \mathbb{E}[A] < \infty $ such that the normalized eigenfunctions $ \psi_E $ of $ H $  almost surely satisfy 
\[
| \psi_E(j) | \, | \psi_E(k) | \leq  A \, 2^{-\mu d(j,k) }
\]
for all $ E \in I $ and $ j,k \in \nn $, cf.\ \cite[Ch.~7]{MR3364516}.

As an aside, we note that by extending an argument of E.~Kritchevski~\cite{MR2352276}, the spectral statement may also be established without Assumption~ \ref{nonconcentration}.
\begin{proposition}[Spectral localization]\label{prop:localization}
The  spectrum of $H $ is almost surely of pure-point type with normalized eigenfunctions satisfying
\begin{equation}\label{eq:Kritloc}
\sum_{k \in \nn} 2^{\frac{c}{4} d(0,k)} |\psi_E(k)|^2 < \infty 
\end{equation}
for any $ E \in \sigma(H) $. 
\end{proposition}
The proof of this result is the subject of Appendix~\ref{Kritchevski}.

Eigenfunction correlator localization~\eqref{eq:ECloc} is a much stronger result than just~\eqref{eq:Kritloc}. In particular, it allows us to draw conclusions about 
the inverse participation ratios (IPRs)
\[
P_q(\psi) \vcentcolon= \frac{\sum_x |\psi(x)|^{2q}}{\left[ \sum_x |\psi(x)|^2 \right]^q } = \frac{\| \psi \|_{2q}^{2q}}{\| \psi \|_2^{2q}}  .
\]
The IPRs are comparable for different values of   $q \geq \frac{1}{2}$, e.g.:
\begin{itemize}
\item for any $ q \geq \frac{1}{2}$:
\begin{equation}
1 \leq P_q(\psi) \, [P_{\frac{q}{2q-1}}(\psi)]^{2q -1} \, ,
\end{equation}
\item for any $ q \geq 1 $ we have 
$ r(\psi) ^{2q} \ 
\le \   P_q (\psi)
\  \le \    r(\psi) ^{2(q-1)} $ where 
 $
r(\psi)  =  \|\psi \|_\infty / \|\psi \|_2   \,  
$
(with $\|\psi\|_\infty  = \sup_{x} |\psi(x) |$). 
\end{itemize}
It therefore remains only to state a result concerning the most prominent case $ q = 2 $. Note that a bound $ P_2(\psi)>\varepsilon^4  $ with $ \varepsilon > 0 $ (independent of $ n $) is a localization statement for $ \psi $.

\begin{corollary}[IPRs]\label{cor:IPR}
 If Assumption~\ref{nonconcentration} is satisfied in a bounded open interval $I \subset \rr$, then there exists some $ C < \infty $ such that for any $ E \in I $ and $ W , \varepsilon > 0 $
 
\begin{equation}\label{eq:IPR}
 \mathbb{P}\left(\begin{array}{c} \mbox{There is $ \psi \in \ell^2(B_n) $ with $ H_n \psi = \lambda \psi $ and} \\ \mbox{$\; | \lambda - E | \leq 2^{- n- 1} W $ such that $\;  P_2( \psi) \leq \varepsilon^4 $} \end{array} \quad \right) \leq  C \, W \varepsilon  
 \end{equation}
 for all $ n \in \nn $.
\end{corollary}
The proof of this corollary, which in fact does not rely on any special structure of the hierarchical model, is spelled out in Appendix~\ref{IPR}. 

In order to appreciate this result, we stress that the smallness of the probability in~\eqref{eq:IPR} is not due to the fact that the interval $ I_n = E +  2^{-n-1}   [ - W ,  W ] $ is typically void of eigenvalues. In fact, as is proven in Theorem~\ref{poissonstatisticsthm} below
\[ 
\lim_{n\to \infty} \mathbb{P}\left( \mbox{No eigenvalue of $ H_n $ in $ I_n $}\right) = \exp\left( -  \nu(E) W \right) 
\]
at all Lebesgue points $ E\in I $ of the (infinite-volume) density of states $ \nu $ (see also~\eqref{eq:DOSm}).

The above corollary is particularly interesting since it challenges the conclusions of~ \cite{PhysRevB.88.045103}, in which numerical findings suggested, for $ c \in (0,1) $ (i.e. $ d_s > 2 $), the existence of an energy  at which the inverse participation ratio $P_2(\psi)$ vanishes for large $ n $ in the case of a weak Gaussian random potential. In fact, the authors of~\cite{PhysRevB.88.045103} study the averaged IPR
\[\Pi_n(I) = \frac{\ee \sum_{\lambda \in \sigma(H_n) \cap I} \|\psi_\lambda\|_4^4}{\ee \sum_{\lambda \in \sigma(H_n) \cap I} 1}
\]
in the limit of vanishingly small intervals $ I \subset \mathbb{R} $.  From~\eqref{eq:IPRproof} in Appendix~\ref{IPR} one concludes
\begin{equation}\label{eq:avgiprbound}
\Pi_n(I) \geq C^{-4} \left( \frac{\nu_n(I)}{|I|} \right)^4  \qquad\mbox{with $ \nu_n(I) = \ee \langle \delta_0, 1_I(H_n) \delta_0\rangle $}
\end{equation}
 for any bounded $ I \subset \mathbb{R} $ in which Assumption~\ref{nonconcentration} is valid.
Since the finite-volume density of states is bounded away from zero for all large enough $ n $ provided the interval $ I $ is strictly contained in the infinite-volume spectrum $ \sigma(H) $,
the right side can be shown to be strictly positive in the limit $ n \to \infty $ (cf.~ \cite{MR2399055}). 
In particular, if $ \varrho $ is a Gaussian distribution, $ \sigma(H) = \mathbb{R} $, and this applies to all energies, which contradicts the conclusions in \cite{PhysRevB.88.045103}.

The second main result of this paper concerns the level statistics of $H$, i.e., the statistical behavior of the rescaled eigenvalues of $H_n$ in the infinite-volume limit $ n \to \infty $. 
More precisely, we show that the random measure 
\[\mu_n(f) = \sum_{\lambda \in \sigma(H_n)} f(2^n(\lambda - E)), \quad \mbox{with $ f \in C_0$,} \]
corresponding to $ E \in \mathbb{R} $ converges to a Poisson point measure with intensity given by the density of states $\nu(E)$ of $H$. The latter is the derivative of the density of states measure 
\begin{equation}\label{eq:DOSm}
\nu(f) = \ee \langle \delta_0, f(H) \delta_0 \rangle \, .
\end{equation}
Note that the Wegner estimate, $|\nu(f) | \leq \| \varrho \|_\infty \| f \|_1 $ (cf.~\cite{MR3364516}), ensures that $ \nu $ is absolutely continuous, i.e., of the form $ \nu(f) = \int f(\lambda) \, \nu(\lambda) \, d\lambda $. 
\begin{theorem}[Poisson statistics]\label{poissonstatisticsthm} Suppose Assumption~\ref{nonconcentration} is satisfied in an open set $I \subset \rr$ and $E \in I$ is a Lebesgue point of $\nu$. Then the rescaled eigenvalue point process $ \mu_n$ 
converges in distribution to a Poisson point process with intensity $\nu(E)$ as $n \to \infty$.
\end{theorem}

Let us conclude this section by summarizing the results above and placing them into context. As mentioned before, the first proof of complete spectral localization for the hierarchical model at any $ c > 0 $ in the special case of a Cauchy probability density is due to  S.~Molchanov \cite{MR1463464}. 
Later, E.~Kritchevski showed \cite{MR2352276,MR2413200} that $H$ almost surely has only pure-point spectrum provided that either $V$ has a Cauchy component or $c > 1/2$ and that the rescaled eigenvalues converge to a Poisson point process if $c > 2$. 
Proposition~\ref{prop:localization} gets rid of any conditions on the distribution or the spectral dimension and may be considered optimal as far as spectral localization is concerned. In
addition, we establish dynamical localization and  information on the IPRs under Assumption~\ref{nonconcentration}  on the renormalized densities  $T_{p_r} ... T_{p_1} \varrho_E$. We believe  this assumption to be 
generally valid irrespective of the spectral dimension and verify it in special cases in Appendix~\ref{densityappendix}. 
Theorem~\ref{poissonstatisticsthm} significantly enlarges the parameter range in which the level statistics are known to converge to a Poisson point process. 
In essence, our results support the conjecture that there is no delocalization transition in the hierarchical Anderson model for any $ c > 0 $, i.e., not even for spectral dimensions $ d_s >2 $.

Another toy model to study effects of dimensions on the Anderson transition has recently been analyzed by C.~Sadel~\cite{Sadel2016}. The antitrees studied there are characterized by a tunable dimension $ d $ which measures surface-to-volume growth rates in these graphs. It is proven that a spectral transition occurs at $ d = 2 $.

\section{The Renormalization Group} \label{renormalizationsection}
The content of this section is the investigation of a relationship between the resolvents of $H$ and the resolvents of an operator $\rg H$ whose parameters $(\p, \varrho)$ have effectively been renormalized to $\mathcal{R} (\p, \varrho)$ (cf.~\eqref{eq:rgdefinition}). This is achieved by considering the new Hamiltonian $\rg H = \rg \Delta + \rg V$ with components
\beq{newlaplacian}\rg \Delta = \sum_{r = 1}^\infty p_{r+1} E_r
\eeq
and
\beq{newpotential}(\rg V)_k = \left(\frac{1}{2V_{2k}} + \frac{1}{2V_{2k+1}} \right)^{-1} +  p_1.
\eeq
The definition~\eqref{newpotential} guarantees that the renormalized potential $\rg V$ consists of independent random variables whose common density is $T_{p_1} \varrho$. 

It will be useful in our proof of Theorem~\ref{localizationthm} to also consider a slightly more general situation, in which the disorder remains independent, but is allowed to have different distributions at different sites. We will thus suppose that $V_k \sim \varrho_k$, where $\{\varrho_k \, | \, k \in \nn\}$ is a collection of probability densities with $\sup_k \|\varrho_k\|_\infty < \infty$. The renormalization transformations~\eqref{newlaplacian} and~\eqref{newpotential} extend to this setting directly, the only difference being that the renormalized potential values $(\rg V)_k$ are now drawn from the densities $ T_{p_1}(\varrho_{2k}, \varrho_{2k+1})$ of the random variables defined in~\eqref{newpotential}.

To obtain a relation between $H$ and $\rg H$, we set $L = 2^{n-1} - 1$ and consider the orthonormal basis $\{e_0, ..., e_L, f_0,...,f_L\}$ of $\ell^2(B_n)$ whose members are
\[e_k = \frac{1}{\sqrt{2}}\left( \delta_{2k} + \delta_{2k+1}\right), \quad f_k = \frac{1}{\sqrt{2}}\left( \delta_{2k} - \delta_{2k+1}\right).\]
Thus $\ell^2(B_n) = E \oplus F$, where $E$ and $F$ are the linear spans of $\{e_0, ..., e_{L}\}$ and $\{f_0, ..., f_{L}\}$, respectively. Let $U_e: \ell^2(B_{n-1}) \to E$ and $U_f: \ell^2(B_{n-1}) \to F$ be the isomorphisms defined by
\[U_e\delta_k = e_k, \quad U_f\delta_k = f_k\]
and let $U  = U_e \oplus U_f$. A direct computation shows that a matrix representation of the form
\beq{hinefbasis} U^\ast H_n U = \begin{pmatrix}(\rg \Delta)_{n-1} + p_1 + V_{ee} & \,\,V_{fe} \\ V_{ef} & \, \, V_{ff} \end{pmatrix}
\eeq
is valid in the site basis of $\ell^2(B_{n-1}) \oplus \ell^2(B_{n-1})$. The entries occurring on the right of~\eqref{hinefbasis} are the operators defined by
\beq{veedef} V_{ee}\delta_k  = V_{ff}\delta_k = \frac{1}{2} (V_{2k} + V_{2k + 1})\delta_k
\eeq
and
\beq{vefdef} V_{fe} \delta_k = V_{ef} \delta_k = \frac{1}{2} (V_{2k} - V_{2k + 1})\delta_k
\eeq
acting between the appropriate factors of $\ell^2(B_{n-1}) \oplus \ell^2(B_{n-1})$. Let us emphasize that~\eqref{veedef} cannot be taken completely literally because $V_{ee}$ maps only the first factor of  $\ell^2(B_{n-1}) \oplus \ell^2(B_{n-1})$ into itself, whereas $V_{ff}$ maps only the second factor into itself. Similar considerations apply to~\eqref{vefdef}.

The Schur complement of $(\rg \Delta)_{n-1} + p_1 +  V_{ee}$ in~\eqref{hinefbasis} is the operator 
\[(\rg \Delta)_{n-1} +  p_1 + V_{ee} - V_{fe}V_{ff}^{-1}V_{ef} = (\rg H)_{n-1}\]
and thus the Schur complement formula for the inverse yields the following proposition.

\begin{proposition} The formula
\[U^*H_n^{-1} U = \begin{pmatrix}1&0\\-V_{ff}^{-1}V_{ef}&1\end{pmatrix}\begin{pmatrix}(\rg H)_{n-1}^{-1}&0\\0&V_{ff}^{-1}\end{pmatrix}\begin{pmatrix}1&-V_{fe}V_{ff}^{-1}\\0&1\end{pmatrix}\]
is valid whenever $H_n, (\rg H)_{n-1}$ and $V_{ff}$ are invertible.
\end{proposition}

We will show in Appendix~\ref{densityappendix} that $T_p(\varrho, \tilde{\varrho}) \in L^\infty$ whenever $\varrho,  \tilde{\varrho} \in L^\infty$ so the Wegner estimate applies to both $H$ and $\rg H$. Therefore, $H_n$ and $(\rg H)_{n-1}$ are almost surely invertible and $V_{ff}$ is almost surely invertible because it is a multiplication operator whose entries are independent continuously distributed random variables. In terms of the operator
\[ S = U_e^* - V_{fe}V_{ff}^{-1} U_f^*,\]
where we have identified $U_e$ and $U_f$ with $U_e \oplus 0$ and $0 \oplus U_f$, respectively, this proves the following important formula.

\begin{corollary}\label{recoveryformula} Let $\varphi, \psi \in \ell^2(B_n)$. Then
\[\langle \varphi, H_n^{-1} \psi \rangle = \langle S\varphi, (\rg H)_{n-1}^{-1} S \psi \rangle + \langle U_f^*\varphi, V_{ff}^{-1} U_f^* \psi \rangle\]
almost surely.
\end{corollary}

We will now use Corollary~\ref{recoveryformula} to bound the fractional moments of the Green function
\[G_n(0,k; E) = \langle \delta_k, (H_n - E)^{-1} \delta_0 \rangle\]
by the fractional moments of its renormalized counterpart
\[\rg G_n(0,k; E) = \langle \delta_k, ((\rg H)_n - E)^{-1} \delta_0 \rangle \, . 
\]
To simplify the analysis, we will restrict ourselves to the case that some of the potential values $V_k$ have a Cauchy distribution
\beq{poissondefn}\varrho_k(v) = P_{\mu + i\sigma}(v) \vcentcolon = \frac{1}{\pi}\frac{\sigma}{(v-\mu)^2 + \sigma^2}.
\eeq
In this case, a decoupling inequality becomes available (see~\cite[Thm. 8.7]{MR3364516}), which states that for every $s \in (0, 1)$ and $z \in \cp$ there exists a constant $D_s(z) \in (0,\infty)$ with the property that
\begin{align}
\frac{1}{D_s(z)} \int \! \frac{1}{|v - \gamma|^s} \, P_z(v) \, dv &\le \int \! \frac{|v|^s}{|v - \gamma|^s} \, P_z(v) \, dv  \label{decoupling1}\\
&\le D_s(z) \int \! \frac{1}{|v - \gamma|^s} \, P_z(v) \, dv  \label{decoupling2}
\end{align}
uniformly in $\gamma \in \cc$. The restriction to the Cauchy case is possible thanks to a partial comparison trick which we will devise in the proof of Theorem~\ref{localizationthm}.

\begin{theorem} \label{fminequality} Let $s \in (0,1)$ and $k \in B_{n-1} \setminus \{0\}$. If
\[\varrho_0 = \varrho_1=\varrho_{2k}=\varrho_{2k+1} = P_z,\]
then both
\[\ee |G_n(0,2k; 0)|^s \le D_s(z)^4\, \ee |\rg G_{n-1}(0,k; 0)|^{s}\]
and
\[\ee |G_n(0,2k + 1; 0)|^s \le D_s(z)^4\, \ee |\rg G_{n-1}(0,k; 0)|^{s}.\]
\end{theorem}
\begin{proof}
We will prove only the estimate for $\ee |G_n(0,2k; 0)|^s$ since the analysis of $\ee |G_n(0,2k + 1; 0)|^s$ then reduces to swapping some indices and changing some signs. The formulas
 \[S\delta_{2\ell} = \frac{1}{\sqrt{2}}\left(1 - \frac{V_{2\ell} - V_{2\ell+1}}{V_{2\ell} + V_{2\ell+1}}\right) \delta_{\ell}\]
and
 \[U_f^\ast \delta_{2\ell} = \delta_\ell\]
 are valid for both $\ell = 0$ and $\ell = k$. Since $k \neq 0$, we necessarily have that $\langle \delta_0, V_{ff}^{-1} \delta_k \rangle = 0$ and hence Corollary~\ref{recoveryformula} asserts that
 \beq{gfrecovery}
 G_n(0, 2k; 0) =   2\frac{V_{1}}{V_{0} + V_{1}} \frac{V_{2k+1}}{V_{2k} + V_{2k+1}} \rg G_{n-1}(0, k; 0).
 \eeq
Consider a term of the form
\[X(\ell)  = \ee_{\ell} \left| \left(\frac{V_{2\ell + 1}}{V_{2\ell} + V_{2\ell+1}}\right) \rg G_{n-1}(0, k; 0) \right|^s,\]
where $\ell \in \{0,k\}$  and $\ee_\ell$ denotes the conditional expectation with respect to $\{V_i \, | \,i \neq 2\ell, 2\ell + 1\}$.
The Green function is of the form
 \[|\rg G_{n-1}(0,k; 0)|^s = \left| \frac{\alpha}{(\rg V)_\ell - \beta}\right|^{s}\]
 for some $\alpha, \beta \in \cc$ which are independent of $(\rg V)_\ell$ (cf.~\cite[Sec.~5.5.]{MR3364516}). Writing
 \[u = \left(\frac{1}{2v} + \frac{1}{2w} \right)^{-1},\]
 it follows that
 \begin{align*}X(\ell) &= \iint  \! \left| \frac{v}{v + w} \frac{\alpha}{u - \beta}\right|^s \, P_z(v)P_z(w) \, dv \, dw\\
 &= \iint  \! \left| \frac{\alpha v}{2vw - \beta(v + w)}\right|^s \,P_z(v)P_z(w) \, dv \, dw\\
 &= \int  \! \left|\frac{\alpha v}{2v - \beta}\right|^s \int \! \left| \frac{1}{w - \beta v (2v - \beta)^{-1} }\right|^s \,P_z(w)\, dw \, P_z(v) \, dv,
 \end{align*}
 where we have absorbed the shift of the renormalized potential by $p_1$ into the constant $\beta$. Applying the decoupling inequality~\eqref{decoupling1} to the inner integral and reversing the previous calculations shows that
 \begin{align*}X(\ell) &\le D_s(z) \, \int  \! \left|\frac{\alpha v}{2v - \beta}\right|^s \int \! \left| \frac{w}{w - \beta v (2v - \beta)^{-1} }\right|^s \,P_z(w)\, dw \, P_z(v) \, dv\\
 &= D_s(z) \, \iint  \! \left| \frac{vw}{v + w} \frac{\alpha}{u - \beta}\right|^s \,P_z(v)P_z(w) \, dv \, dw\\
 &= 2^{-s} D_s(z) \iint  \! \left|  \frac{\alpha u}{u - \beta}\right|^s \, P_z(v)P_z(w) \, dv \, dw\\
 &= 2^{-s} D_s(z) \int \! \left|  \frac{\alpha u}{u - \beta}\right|^s \, (T_0P_z)(u) \, du.
 \end{align*}
 It is easy to see that $T_0P_z = P_z$ (see also Appendix~\ref{densityappendix}), so applying the decoupling inequality~\eqref{decoupling2} yields
 \begin{align*}X(\ell) &\le 2^{-s}D_s(z)^2 \int  \! \left|  \frac{ \alpha}{u - \beta}\right|^s \, (T_0 P_z)(u)\, du\\
 &= 2^{-s}  D_s(z)^2\iint  \! \left| \frac{\alpha}{u - \beta}\right|^s \, P_z(v)P_z(w) \, dv \, dw\\
 &= 2^{-s}  D_s(z)^2 \, \ee_{\ell} |\rg G_{n-1}(0, k; 0)|^s.
 \end{align*}
 Since $V_0$ and $V_1$ are independent of $V_{2k}$ and $V_{2k+1}$, combining the bound for $X(\ell)$ with~\eqref{gfrecovery} implies
 
 \begin{align*}
 \ee |G_n(0, 2k, 0)|^s &= 2^s \, \ee  \left( \left|\frac{V_{1}}{V_{0} + V_{1}} \right|^s X(k) \right)\\
 &\le  D_s(z)^2 \, \ee \left( \left|\frac{V_{1}}{V_{0} + V_{1}} \right|^s \ee_k \left| \rg G_{n-1}(0, k; 0) \right|^s \right)\\
 &= D_s(z)^2  \, \ee X(0)\\
 &\le D_s(z)^4 \, \ee \left| \rg G_{n-1}(0, k; 0) \right|^s.  
 \end{align*}
\end{proof}

Along with the restricted operators $H_n$, there is another sequence of truncations
\begin{equation}\label{def:Hnm}
H_{n,m} = 1_{B_n} \left(\sum_{r=1}^m p_rE_r + V \right) 1_{B_n},
\end{equation}
which will be useful in our proof of Theorem~\ref{poissonstatisticsthm}. Notice that
\[H_n = H_{n,n} + \alpha \ketbra{\varphi_n}\]
with $\alpha = \sum_{r > n} 2^{n-r} p_r$, so reasoning analogous to Corollary~\ref{recoveryformula} shows that also
\beq{recoveryformula2}\langle \varphi, H_{n,n}^{-1} \psi \rangle = \langle S\varphi, (\rg H)_{n-1,n-1}^{-1} S \psi \rangle + \langle U_f^*\varphi, V_{ff}^{-1} U_f^* \psi \rangle
\eeq
almost surely.  The formula~\eqref{recoveryformula2} lets us determine the distribution of the quantity
\[\frac{1}{\Phi_n(E)} = \langle \varphi_n, (H_{n,n} - E)^{-1} \varphi_n \rangle \]
explicitly in terms of the operators $T_p$ when the the disorder has the same distribution at each site, that is, $\varrho_k = \varrho$ for every $k \in \nn$.
\begin{corollary} \label{fdistribution} The density of $\Phi_n(E)$ is given by $T_{p_n}...T_{p_1} \varrho_E$.
\end{corollary}
\begin{proof}
 Notice that $S\varphi_n = \varphi_{n-1}$ and $U^\ast_f\varphi_n = 0$ so~\eqref{recoveryformula2} shows that
 \[\frac{1}{\Phi_n(0)} = \langle \varphi_{n-1}, (\rg H)_{n-1,n-1}^{-1} \varphi_{n-1} \rangle = \frac{1}{\rg \Phi_{n-1}(0)}.\]
 We can continue renormalizing in this fashion until we reach a Hamiltonian consisting of a $1 \times 1$ random matrix whose element is distributed as $T_{p_n}...T_{p_1}\varrho$. This proves the result for $E = 0$. The general case follows by shifting the density of the original potential by $-E$.
\end{proof}

\section{Proof of Localization}\label{localizationsection}
We begin our proof of Theorem~\ref{localizationthm} by considering Hamiltonians with single-site densities $\{\varrho_k \, | \, k \in \nn\}$ which may vary from site to site, and proving a uniform high-disorder bound for the Green function in terms of the relative strengths of the hopping $|p_r| \le \epsilon 2^{-cr}$ and the disorder $\sup_{k \in \nn} \|\varrho_k\|_\infty$.

\begin{proposition}\label{highdisorder}
If $s \in (0,1)$ and $\mu > 0$ satisfy
\[1 + \mu < s(1+c),\]
then there exist $ \epsilon_0 > 0$ and $C \in ( 0,\infty)$ such that
\[\sup_{k \in \nn} 2^{(1 + \mu)d(0,k)} \left( \sup_{n \geq 1} \ee |G_n(0, k; 0)|^s\right) \le C\]
for any collection of single-site densities satisfying $\epsilon \left( \sup_{i \in \nn} \|\varrho_i\|_\infty \right)< \epsilon_0$.
\end{proposition}
\begin{proof}
Since our method of proof is completely standard, and every detail of the argument can be found in a general setting in~\cite[Ch. 10]{MR3364516}, we provide only a sketch of the proof. Let $G_\Lambda$ denote the Green function of the restriction of $H$ to a finite volume $\Lambda \subset \nn$. If $k \neq 0$, deleting the matrix elements $\{\Delta(k, j), \Delta(j,k) \, | \, j \neq k\}$ from $\Delta$ and applying the resolvent identity yields the formula
\[G_{\Lambda}(0,k; E) =  - \sum_{j \neq k} G_{\Lambda }(k,k; E) \Delta(k,j) G_{\Lambda \setminus \{k\}}(0,j; E).\]
Let $M = \sup_{i \in \nn} \|\varrho_i\|_\infty$. Factoring the expectation through the conditional expectation $ \ee_k $ with respect to $V_k$, we obtain
\beq{glambdabound}
  \ee |G_{\Lambda}(0,k; E)|^s \le  \frac{ 2M^s}{1-s}\sum_{j \neq k} |\Delta(k,j)|^s \ee |G_{\Lambda \setminus \{k\}}(0,j; E)|^s,
\eeq
because $G_{\Lambda \setminus \{k\}}(0,j, E)$ does not depend on $V_k$ and
\[\ee_k |G_{\Lambda }(k,k; E)|^s \le \frac{2\|\varrho_k\|_\infty^s}{1-s}.\]
Setting
\[f(k) = \sup_{|\Lambda| < \infty} \ee |G_\Lambda(0, k; 0) |^s < \infty\]
and taking the supremum over all finite $\Lambda \subset \nn$ in~\eqref{glambdabound} yields
\beq{subharmonicityrelation}
  f(k) \le AM^s\left(\delta_{0,k} + \sum_{j \neq k} |\Delta(k,j)|^{s} f(j) \right)
\eeq
with $A = 2/(1-s)$. Thus~~\eqref{eq:LaplaceME} and the inequality  $1 + \mu < s(1+c)$ show that
\begin{align*}  AM^s\sup_{k \in \nn} \sum_{j \in \nn}\frac{2^{\mu d(0, j)}}{2^{\mu d(0,k)}} |\Delta(k,j)|^{s} &\le A M^s \sum_{j \in \nn}2^{\mu d(0, j)}|\Delta(0,j)|^{s}\\
&\le A^\prime\epsilon^s M^s \sum_{j \in \nn}2^{\mu d(0, j)}2^{-s(1+c) d(0,j)}\\
&= A^\prime\epsilon^s M^s\sum_{j \in \nn}2^{( \mu -s(1+c)) d(0,j)}\\
&\le A^\prime\epsilon^s M^s\sum_{r \geq 0} 2^{( 1 + \mu -s(1+c)) r} < 1
\end{align*}
provided $\epsilon M < \epsilon_0$ is small enough. Hence, by iterating~\eqref{subharmonicityrelation}, 
\[C = \sum_{k \in \nn} 2^{\mu d(0,k)} f(k) < \infty\]
which implies
\[\sup_{n} \sum_{k \in \nn} 2^{\mu d(0,k)} \ee |G_n(0,k; 0)|^s \le C. \]
The theorem now follows by observing that $|B_r \setminus B_{r-1}| = 2^{r-1}$ for all $r \geq 1$ and that $\ee |G_n(0,k; 0)|^s$ depends on $k$ only in terms of $d(0,k)$ .
\end{proof}

We will now return to the setting of Theorem~\ref{localizationthm} in which the potential was identically distributed with a common density $\varrho$. Our strategy is to extend the conclusion of Theorem~\ref{highdisorder} to the entire parameter range by renormalizing into the high-disorder regime. This is based on the observation that, when $|p_r| \le \epsilon 2^{-cr}$, the renormalized hopping $\rg \p$ satisfies
\beq{observation} |(\rg \p)_r| \le 2^{-c}\epsilon 2^{-cr} \approx 2^{-c}|p_r|
\eeq
so that the renormalization has effectively decreased $\epsilon$ by a factor $2^{-c}$.

\begin{theorem} Suppose Assumption~\ref{nonconcentration} is true in a bounded interval $I \subset \rr$. If $s \in (0,1)$ and $\mu > 0$ satisfy
\[1 + \mu < s(1 + c),\]
then there exists $C < \infty$ such that
\[\sup_{n \geq 1} \ee |G_n(0, k; E)|^s \le C \, 2^{-(1 + \mu)d(0,k)}\]
for all $k \in \nn$ and $E \in I$.
\end{theorem}

\begin{proof} Since $I$ is bounded, the requirement~\eqref{rhodecay} means that there exist $z \in \cp$ and $C_I < \infty$ such that
\beq{loc:cauchydom}\varrho_E(v) \le C_I P_z(v)
\eeq
for all $E \in I$ and $v \in \rr$, where $P_z$ is the Poisson kernel defined in~\eqref{poissondefn}. The following bound for $G_n(0,k; 0)$ will depend only on $z$, $C_I$, $\|\varrho\|_\infty$, and the constants occuring in Assumption~\ref{nonconcentration}, which implies that we can restrict ourselves to the situation where $E = 0 \in I$ without any loss of generality.

Suppose $n \geq N \geq 1$ and let $k \in B_n$. We will first consider the Hamiltonian $H^\prime = \Delta + V^\prime$ which is obtained from $H$ by replacing the potential values in $B_N(0) \cup B_N(k)$ by random variables with the Cauchy distribution $P_z$. Thus $V^\prime_i$ has the density
\[\varrho_i = \begin{cases} P_z & \mbox{ if } i \in B_N(0) \cup B_N(k)\\ \varrho & \mbox{ else } \end{cases}.\]
Since $T_pP_z= P_{z+p}$, the renormalized potential $\rg^N V^\prime$ has densities
\[ \varrho_i = \begin{cases} P_{z + p_1 + ... + p_N} & \mbox{ if } i \in \left\{0, \lfloor 2^{-N}k \rfloor \right\}\\ T_{p_N}...T_{p_1}\varrho & \mbox{ else } \end{cases},\]
and by iterating the observation~\eqref{observation}, $\rg^N \Delta$ has a hopping strength
\[ \left| ( \rg^N \p)_r \right| = \left| p_{r+N} \right| \le \epsilon_N 2^{-cr}, \quad \epsilon_N = 2^{-c N}\epsilon.\]
Because $\varrho$ satisfies Assumption~\ref{nonconcentration} in $I$ and $\|P_{z + p}\|_\infty \le (\im z)^{-1}$ for all $p \in \rr$, this implies that the hypothesis
\[ \epsilon \left( \sup_{i \in \nn} \|\varrho_i\|_\infty \right)< \epsilon_0\]
of Proposition~\ref{highdisorder} is eventually satisfied by $\rg^N H^\prime$ for some sufficiently large $N$ which depends on $z$ and the constants in Assumption~\ref{nonconcentration}. Hence, when $1 + \mu < s (1 + c)$,  there is some $C_0 < \infty$ such that
\[\sup_{n \geq 1} \ee \left|\rg^N G_n^\prime \left(0, \lfloor 2^{-N}k \rfloor; 0 \right)\right|^s \le C_0 \, 2^{-(1 + \mu) d\left(0, \lfloor 2^{-N}k \rfloor \right) },\]
where $G_n^\prime$ denotes the Green function of $H_n^\prime$. If $k \in B_n \setminus B_N$, then $N$ successive applications of Theorem~\ref{fminequality} show that
\beq{cauchygboundnoBN} \ee |G_n^\prime(0, k; 0)|^s \le D\, C_0  \, 2^{-(1 + \mu) d\left(0, \lfloor 2^{-N}k \rfloor \right) } \le D \, C_0  \, 2^{-(1 + \mu) (d\left(0, k \right) - N) },
\eeq
with
\[D =  \left[D_s(z)D_s(z + p_1) ... D_z(z + p_1 + ...p_N)\right]^4.\]
Since $H^\prime$ is obtained from $H$ by replacing $\{V_i \, | \, i \in B_N(0) \cup B_N(k)\}$ with random variables distributed according to $P_z$,~\eqref{loc:cauchydom} and~\eqref{cauchygboundnoBN} show that 
\beq{realgbound} \ee |G_n(0,k; 0)|^s \le C_I^{2|B_N|}   \ee |G^\prime_n(0,k; 0)|^s \le C_1   \, 2^{-(1 + \mu) d\left(0, k \right) }
\eeq
for some $C_1 < \infty$ which depends on $z$, $C_I$, and the constants occurring in Assumption~\ref{nonconcentration}. If $k \in B_N$, then the a priori bound
\[\ee |G_n(0,k;0)|^s \le \frac{4\|\rho\|_\infty^s}{1-s}\]
is valid so~\eqref{realgbound} implies that
\[\ee |G_n(0, k; 0)|^s \le C 2^{-(1+\mu) d(0,k)} \]
with a constant $C < \infty$ depending only on $z$, $C_I$, $\|\varrho\|_\infty$, and the constants occurring in Assumption~\ref{nonconcentration}.
\end{proof}

Theorem~\ref{localizationthm} is a consequence of the relationship between eigenfunction correlators and Green functions.

\begin{proof}[Proof of Theorem~\ref{localizationthm}] 
Theorem~\ref{localizationthm} follows from the standard result (see Chapter 7 in~\cite{MR3364516})
\[\ee Q_n(j, k; I) \le  C_s \, \ee \int_I \! |G_n(j,k; E) |^s \, dE\]
and the fact that $\ee |G_n(j,k; E)|^s$ depends on $j$ and $k$ only in terms of $d(j,k)$.
\end{proof}

\section{Proof of Poisson Statistics} \label{poissonsection}
This section is devoted to the proof of Theorem~\ref{poissonstatisticsthm} concerning the convergence of the random point measure
\[\mu_n(f) = \sum_{\lambda \in \sigma(H_n)} f(2^n(\lambda - E))\]
to a Poisson point process with intensity $\nu(E)$ when $E$ is a Lebesgue point of the density of states. In this setting, $H$ was a Hamiltonian with fixed hopping $|p_r| \le \epsilon 2^{-cr}$ and a single-site density $\varrho \in L^\infty$ such that Assumption~\ref{nonconcentration} is valid in a neighborhood of $E$ for some $\delta > 0$. Our argument is based on the following fundamental fact~\cite[Prop. 17.5]{MR3364516}, which essentially characterizes Poisson point processes as simple point processes consisting of infinitely many independent components.

\begin{proposition} \label{poissoncharacterization}
  Consider a sequence of point processes of the form $\mu_n = \sum_j \mu_{n,j}$, where $\{\mu_{n,j} \, | \, j = 1, ..., N_n\}$ is a triangular array of point processes with the following properties:
  \begin{enumerate}[i.]
    \item The point processes $\{\mu_{n,1}, ..., \mu_{n, N_n}\}$ are independent for all $n \geq 1$.
    \item If $B \subset \rr$ is a bounded Borel set, then
    \[\lim_{n \to \infty} \sup_{j \le N_n} \pp(\mu_{n,j}(B) \geq 1) = 0.\]
    \item There exists some $c \geq 0$ such that if $B \subset \rr$ is a bounded Borel set with $|\partial B| = 0$, then
    \[\lim_{n \to \infty} \sum_{j=1}^{N_n} \pp(\mu_{n,j}(B) \geq 1) = c|B|\]
    and
    \[\lim_{n \to \infty} \sum_{j=1}^{N_n} \pp(\mu_{n,j}(B) \geq 2) = 0 .\]
  \end{enumerate}
Then $\mu_n$ converges in distribution to a Poisson point process with intensity~$c$.
\end{proposition}
  Among the several equivalent options available~\cite{MR950166,MR1876169}, we choose the definition that a sequence of point processes $\mu_n$ converges in distribution to $\mu$ provided
\[\lim_{n \to \infty} \ee e^{-\mu_n(P_z)}= \ee e^{-\mu(P_z)}\]
for all $z \in \cp$,  where $P_z$ is the Poisson kernel~\eqref{poissondefn}. Hence, Theorem~\ref{poissonstatisticsthm} can be established by finding a sequence $\tilde{\mu}_n$ such that Proposition~\ref{poissoncharacterization} applies to $\tilde{\mu}_n$ and
\[ \lim_{n \to \infty} \ee e^{-\tilde{\mu}_n(P_z)} = \lim_{n \to \infty} \ee e^{-\mu_n(P_z)}\]
for all $z \in \cp$. The truncated operators $H_{n,m}$ (cf.~\eqref{def:Hnm}) provide a valuable tool in this endeavor because, for any $m \le k \le n$,
\beq{hnmdecomposition}
H_{n,m} = \bigoplus_{j=1}^{2^{n-k}} H_{k,m}^{(j)} ,
\eeq
and each $H_{k,m}^{(j)}$ is an independent copy of $H_{k,m}$. The relationship between $H_n, H_{n,n}$, and $H_{n,n-1}$ is essentially controlled by the quantity featured in the next lemma.
\begin{lemma}\label{fnlemma} Let  
\[F_n(z) := \langle \varphi_n, (H_{n,n-1} - z)^{-1} \varphi_n\rangle\]
with $\varphi_n = 2^{-n/2}1_{B_n}$ and $ z \in \mathbb{C}_+ $. Then:
  \begin{enumerate}[i.]
 \item $\varphi_n$ is almost surely cyclic for $H_{n,n-1}$.
 \item If Assumption~\ref{nonconcentration} holds for $ I \subset \mathbb{R} $ then there exists $C < \infty$ such that 
\[
 \pp( \left| F_n(t) \right|  \ge  |\alpha|^{-1})  \le C \, 2^{(c-\delta)n} \, |\alpha|  \]
 for all $t \in I$ and $ \alpha \neq 0 $.
  \end{enumerate}
\end{lemma}
\begin{proof} The vector $\varphi_n$ is cyclic for $H_{n,n-1}$ if and only if
\[{\rm span}\{f(H_{n,n-1}) \varphi \, | \, f \in C_0\} = \ell^2(B_n),\]
which is clearly true almost surely when $n = 1$. Now suppose the result is true for $H_{n,n-1}$. Since
\[H_{n,n} = H_{n,n-1} + p_n \ketbra{\varphi_n},\]
$\varphi_n$ is cyclic for $H_{n,n}$ whenever it is cyclic for $H_{n,n-1}$ ~\cite{MR2154153}. It follows that $\varphi_{n+1} = \frac{1}{\sqrt{2}}(\varphi_n \oplus \varphi_n)$ is cyclic for $H_{n+1,n} = H_{n,n}^{(1)} \oplus H_{n,n}^{(2)}$ when the spectrum is simple, as is almost surely the case by the Minami estimate~\cite{Min96} (see also~\cite{MR3364516}).

 For the second part, recall Lemma~\ref{fdistribution}, which asserts that 
\[\Phi_{n-1}(t) = \left(\langle \varphi_{n-1}, (H_{n-1,n-1} - t)^{-1} \varphi_{n-1} \rangle \right)^{-1}\]
is a random variable with density $T_{p_{n-1}}...T_{p_1} \varrho_t$. Since $F_n(t)$ is an average of two independent copies of $\left(\Phi_{n-1}(t)\right)^{-1}$, we have: 
\begin{align*} \pp( \left| F_n(t) \right|  \ge  |\alpha|^{-1}) &\le 2 \, \pp( \left| \Phi_{n-1}(t)\right| \leq |\alpha| )\\
& \le  4 |\alpha| \,  \|T_{p_{n-1}}...T_{p_1} \varrho_t\|_\infty \\
&  = |\alpha| \cdot  \mathcal{O}\left(2^{(c-\delta)n}\right) 
\end{align*}
where the last estimate holds uniformly in $t \in I$ thanks to Assumption~\ref{nonconcentration}.
\end{proof}

Our next goal is to understand how the finite-volume density of states
\[\nu_n(f) = 2^{-n} \tr f(H_n)\]
is approximated by its analogue
\[\nu_{n,m}(f) = 2^{-n} \tr f(H_{n,m}),\]
which is the content of Theorem~\ref{approximationbynunm} below. For its statement, we introduce the notation
\[\zel \vcentcolon= E + 2^{-\ell} z\]
for all $z \in \cp$ and $\ell \geq 0$. The connection between Theorem~\ref{poissonstatisticsthm} and $\nu_n$ is through the formula
\beq{munupoisson}
\mu_n(P_z) = \nu_n(P_{z_n}).
\eeq

\begin{theorem}\label{approximationbynunm}
Suppose Assumption~\ref{nonconcentration} is satisfied in an open set $I \subset \rr$ and $E \in I$. 
Let $z \in \cp$ and set
$ o_\ell(z) := \int_{I^c } P_{\zel}(t) dt $. Then:
 \begin{enumerate}[i.]
 \item There is some $ \varepsilon > 0 $ such that $  o_\ell(z)  \leq \frac{2 \im z}{\varepsilon \, \pi} \, 2^{-\ell } $. In particular, 
$ o_\ell(z) $ is a null sequence for any $ z \in \cp $ as $ \ell \to \infty $. 
\item There is some $C < \infty$, which does not depend on $n,m, \ell $ or $z $, such that for all $ m \leq n $:
\begin{align} \label{nunmbound1}
& \ee |\nu_{n,n}(P_{\zel}) - \nu_{n,m}(P_{\zel})|  \le C \, (\im z)^{-1}  \, 2^{\ell -m } \left( 2^{-\delta m}  +  o_\ell(z)  \right) \, , \\
\label{nunmbound2}
& \ee |\nu_n(P_{\zel}) - \nu_{n,n}(P_{\zel})| \le C \,  (\im z)^{-1}   \, 2^{\ell -n }   \left( 2^{-\delta n}  + o_\ell(z) \right) \, . 
\end{align}
  \end{enumerate}
\end{theorem}
\begin{proof}
The first assertion follows from the fact that there is a $ \varepsilon $-neighborhood of $ E \in I $ which is fully contained in $ I $ together with a simple explicit computation.

For a proof of the second assertion we set $\alpha =  p_n$ so that
\[H_{n,n} = H_{n,n-1} + \alpha \ketbra{\varphi_n}.\]
Since $\varphi_n$ is almost surely cyclic for $H_{n,n-1}$, the theory of rank-one perturbations~\cite{MR2154153} shows that the following statements are valid:
\begin{itemize}
\item The eigenvalues of $H_{n,n-1}$ coincide with the set of poles of $F_n$.
\item The eigenvalues of $H_{n,n }$ coincide with the set $\{E \in \rr \, | \, F_n(E) = -\alpha^{-1}\}$.
\item The function $F_n$ is monotone increasing between its poles.
\end{itemize}
For the sake of clarity, let us spell out the proof only in case $ \alpha > 0 $ (the case $ \alpha < 0 $ being similar). Setting $W = \{t \in \rr \, | \, F_n(t)  \le -  \alpha^{-1} \}$, the fundamental theorem of calculus implies that 
\[\nu_{n,n}(P_\zel) - \nu_{n,n-1}(P_\zel) = 2^{-n}  \int \!  1_W(t)  P^\prime_\zel (t) \, dt.\]
Since $\left| P^\prime_z(t) \right| \le (\im z)^{-1} P_z(t)$, taking the expected value yields
\begin{multline*} \ee \left| \nu_{n,n}(P_\zel) - \nu_{n,n-1} (P_\zel) \right| \le  2^{\ell - n} (\im z)^{-1} \,\ee \! \int \! 1_W(t)  P_\zel(t)\, dt\\
 = 2^{\ell - n} (\im z)^{-1} \int \! P_\zel(t)\,  \pp(F_n(t)  \le -  \alpha^{-1}) \, dt.
\end{multline*}
Because $\alpha = \mathcal{O}(2^{-c n})$, Lemma~\ref{fnlemma} asserts that
$ \pp(F_n(t)  \le -  \alpha^{-1}) \le C \, 2^{-\delta n} $
for all $t\in I $ so that
\beq{limsup}
 \int \! P_\zel(t) \, \pp(F_n(t)  \le -  \alpha^{-1}) \, dt \le C \, 2^{-\delta n} + o_l(z) .
\eeq
This proves~\eqref{nunmbound1} when $m = n-1$. Moreover, setting $\alpha = \sum_{r > n} 2^{n-r}p_r = \mathcal{O}(2^{-cn})$ and repeating the argument above with $\nu_n$ in place of $\nu_{n,n}$ and $\nu_{n,n}$ in place of $\nu_{n,n-1}$ proves~\eqref{nunmbound2}.

 For a proof of~\eqref{nunmbound1}, we expand in a  telescopic sum, i.e., for any $f \in C_0$ 
\[\nu_n(f) - \nu_{n,m}(f) = \nu_n(f) - \nu_{n,n-1}(f) +  \sum_{k=m+1}^{n-1} \left( \nu_{n, k}(f) - \nu_{n,k-1}(f) \right)\]
and
\begin{align} \label{telescopictrace} \nu_{n, k}(f) - \nu_{n,k-1}(f) &= 2^{-n} \left(\tr f(H_{n,k}) - \tr f(H_{n,k-1}) \right) \nonumber  \\
&= 2^{-(n-k)} \sum_{j=1}^{2^{n-k}} 2^{-k} \left(\tr f(H_{k,k}^{(j)}) - \tr f(H_{k,k-1}^{(j)}) \right) \nonumber \\
&= 2^{-(n-k)} \sum_{j=1}^{2^{n-k}}\left(\nu_{k,k}^{(j)}(f) -\nu_{k,k-1}^{(j)}(f)\right)
\end{align}
because of the decomposition~\eqref{hnmdecomposition}.
Taking moments and noticing that each term in~\eqref{telescopictrace} has the same distribution yields 
\begin{align*}\ee|\nu_n(P_\zel) - \nu_{n,m}(P_\zel)| \le & \  C \,  (\im z)^{-1}   \, 2^{\ell -n }   \left( 2^{-\delta n}  + o_\ell(z) \right) \\ 
\ \quad & +  \sum_{k=m+1}^{n-1}C \,  (\im z)^{-1}   \, 2^{\ell - (k-1)}   \left( 2^{-\delta (k-1)}  + o_\ell(z) \right) \\
\le & \  C \, (\im z)^{-1}  \, 2^{\ell -m } \left( 2^{-\delta m}  +  o_\ell(z)  \right) . 
\end{align*}
\end{proof}

By the Wegner estimate, the measures $\ee \nu_n$ and $\ee \nu_{n,m}$ are absolutely continuous with densities that are uniformly bounded independently of $n$ and $m$, and~\eqref{munupoisson} shows that the same is true of $\ee \mu_n$. Moreover, by ergodicity ~\cite{MR3364516,MR2413200,MR2909106},
\begin{equation}\label{eq:ergod}
\lim_{n \to \infty} \ee \nu_n(f) = \nu(f)
\end{equation}
for all $f \in C_0$. We will now show that this limit also exists with $\nu_n$ replaced by $\mu_n$.
\begin{corollary}\label{intensity} If $E$ is a Lebesgue point of $\nu$, then
 \[\lim_{n \to \infty} \ee \mu_n(B) = \nu(E)|B|\]
 for all bounded Borel sets $B\subset \mathbb{R} $.
\end{corollary}
\begin{proof}
  That $E$ is a Lebesgue point of $\nu$ means that
  \[\lim_{n \to \infty} 2^n \nu(2^{-n}B + E) = \nu(E)|B|\]
  so it suffices to prove the relation
  \beq{borelsetlimit}
  \lim_{n \to \infty} \left( \ee \mu_n(B) - 2^n \nu(2^{-n}B + E) \right)= 0.
  \eeq
  Since $B$ is bounded, $1_B$ can be approximated arbitrarily well in $L^1$ by finite linear combinations from the set $\{P_z \, | \, z \in \cp\}$. Moreover, since the measures occurring in~\eqref{borelsetlimit} are absolutely continuous with densities bounded uniformly in $n$, we conclude that it is enough to show~\eqref{borelsetlimit} with $B$ replaced by $P_z$. By~\eqref{munupoisson}, this is equivalent to 
  \[\lim_{n \to \infty} \left( \ee \nu_n(P_{z_n}) - \nu(P_{z_n}) \right) = 0.\]
  Applying~\eqref{eq:ergod} and the fact that $ \ee \nu_{p,n} = \ee \nu_{n,n} $ for any $ p \geq n $ (cf.~\eqref{hnmdecomposition}) we conclude from Theorem~\ref{approximationbynunm}  that
  \begin{align*}\lim_{n \to \infty} \left|\ee \nu_n(P_{z_n}) - \nu(P_{z_n})\right| &= \lim_{n \to \infty}\lim_{p \to \infty} \left|\ee \left[ \nu_n(P_{z_n}) - \nu_p(P_{z_n}) \right]\right|\\
   &= \lim_{n \to \infty}\lim_{p \to \infty} \left| \ee \left[ \nu_{p,n}(P_{z_n}) - \nu_p(P_{z_n}) \right]\right| = 0. 
  \end{align*}
\end{proof}

  The next corollary defines the approximating processes $\tilde{\mu}_n$ alluded to earlier.\begin{corollary}\label{divisibility}
  There exists a sequence $m_n$ with $m_n \to \infty$ and $0 < n - m_n \to \infty$ such that the measure defined by
  \[\tilde{\mu}_n(P_z) = \nu_{n,m_n}(P_{z_n})\]
  satisfies
  \[\lim_{n \to \infty} \ee |\mu_n(P_z) - \tilde{\mu}_n(P_z)| = 0\]
  for all $z \in \cp$.
\end{corollary}
\begin{proof}
 Using~\eqref{munupoisson} and Theorem~\ref{approximationbynunm} we see that
 \begin{multline*}
 \ee |\mu_n(P_z) - \tilde{\mu}_n(P_z)| \  = \ee |\nu_n(P_{z_n}) - \nu_{n,m_n}(P_{z_n})| \\
  \le C (\im z)^{-1}  \left[  \, 2^{-\delta n} + \, 2^{n - (1+\delta)m_n} +  \frac{2 \im z}{\varepsilon \, \pi} \left( 2^{-n} + 2^{-m_n} \right) \right].
 \end{multline*}
 Since $\delta > 0$, we can choose $m_n$ such that $m_n \to \infty$, $n - m_n \to \infty$ and $n - (1+\delta)m_n \to -\infty$ which proves the result.
\end{proof}

 Combining the fact that $|e^{-t_1} - e^{-t_2}| \le |t_1 - t_2|$ when $t_1, t_2 \geq 0$ with Corollary~\ref{divisibility} implies that $\tilde{\mu}_n$ satisfies
\[ \lim_{n \to \infty} \ee e^{-\tilde{\mu}_n(P_z)} = \lim_{n \to \infty} \ee e^{-\mu_n(P_z)}.\]
It thus remains to show that $\tilde{\mu}_n$ satisfies the hypothesis of Proposition~\ref{poissoncharacterization}. In the interest of readability, we will suppress the dependence on $n$ and write simply $m$ in place of $m_n$ for the remainder of this section. By~\eqref{hnmdecomposition}, $\tilde{\mu}_n$ is a sum of independent point processes
\[\tilde{\mu}_n = \sum_{j=1}^{2^{n-m}} \tilde{\mu}_{n,j}\]
with
\[\tilde{\mu}_{n,j}(B) = \tr 1_{2^{-n}B + E}(H_{m,m}^{(j)})\]
for all Borel sets $B \subset \rr$. By a theorem of Combes-Germinet-Klein~\cite{CGK09} (cf.~\cite{MR3364516,Min96}), 
\[\pp(\tr 1_B(H_{m,m}) \geq \ell) \le \frac{\left(C \, 2^m |B|\right)^\ell}{\ell!} \]
which implies
\begin{equation}\label{cgkbound}
  \pp(\tilde{\mu}_{n,j}(B) \geq \ell) \le \frac{(C |B| \, 2^{m-n})^\ell}{\ell!}.
\end{equation}
Since $n-m \to \infty$, this shows immediately that the first requirement of Proposition~\ref{poissoncharacterization} is satisfied. For the other requirements, let us abbreviate
\[X(n, \ell) = \sum_{j=1}^{2^{n-m}} \pp(\tilde{\mu}_{n,j}(B) \geq \ell)\]
so that~\eqref{cgkbound} implies
\[X(n,\ell) \le 2^{n-m} \frac{(C |B|\, 2^{m-n})^\ell}{\ell!} \to 0\]
when $\ell \geq 2$. In particular, $X(n, 2) \to 0$ and the last assumption of Proposition~\ref{poissoncharacterization} is satisfied. Since $\tilde{\mu}_{n,j}(B)$ takes values in the non-negative integers
\[\lim_{n \to \infty} X(n,1) = \lim_{n \to \infty} \sum_{j=1}^{2^{n - m}} \ee \tilde{\mu}_{n,j}(B) - \lim_{n \to \infty} \sum_{\ell \geq 2} X(n,\ell).\]
By~\eqref{cgkbound} and the dominated convergence theorem
\[\lim_{n \to \infty} X(n,1) = \lim_{n \to \infty} \sum_{j=1}^{2^{n - m}} \ee \tilde{\mu}_{n,j}(B) = \lim_{n \to \infty} \ee  \tilde{\mu}_{n}(B),\]
so to finish the proof it suffices to derive the identity
\begin{equation}\label{intensityconvergence}
 \lim_{n \to \infty} \ee  \tilde{\mu}_{n}(B) = \nu(E)|B|
\end{equation}
for any bounded Borel set $B$. Since $\tilde{\mu}_{n}(B)$ and $\mu_n$ are easily seen to have uniformly bounded densities, we can approximate $1_B$ by linear combinations from $\{P_z \, | \, z \in \cp\}$ and use Corollary~\ref{divisibility} to see that
\[\lim_{n \to \infty} \ee |\tilde{\mu}_{n}(B) - \mu_n(B)| = 0\]
for all bounded Borel sets $B$. Thus we can replace $\ee \tilde{\mu}_{n}(B)$ with $\ee \mu_n(B)$ in~\eqref{intensityconvergence} and Corollary~\ref{intensity} concludes the proof of Theorem~\ref{poissonstatisticsthm}.

\appendix
\section{The Renormalized Density}\label{densityappendix}
This appendix consists of the proofs of several claims regarding the renormalized densities $T_{p_r}...T_{p_1}\varrho_E$ made in the main part of the text. Let us start by proving the claim in Section~\ref{renormalizationsection}, that $T_p(\varrho, \tilde{\varrho})$ is bounded if $\varrho$ and $\tilde{\varrho}$ are.

\begin{lemma} \label{easybound} Suppose $\varrho, \tilde{\varrho}  \in L^\infty$ are probability densities. Then
\[\|T_p(\varrho, \tilde{\varrho})\|_\infty \le \|\varrho\|_\infty + \|\tilde{\varrho}\|_\infty \]
for any $p \in \rr$.
\end{lemma}
\begin{proof} Notice that
\[\ee_{T_0\varrho} f  \vcentcolon = \int \! f(v) T_0(\varrho, \tilde{\varrho})(v) \, dv = \int \! f \left(\frac{2vw}{v + w}\right) \varrho(v) \tilde{\varrho}(w) \, dv \, dw\]
for any sufficiently regular $f \in L^1$ and
\[\left(\frac{\partial}{\partial v} + \frac{\partial}{\partial w}\right) \left(\frac{2vw}{v + w}\right) = 2\frac{v^2 + w^2}{(v+w)^2} \geq 1.\]
Thus $\ee_{T_0\varrho} f$ is bounded by
\begin{align*}
 & \int \! f \left(\frac{2vw}{v + w}\right) \varrho(v) \tilde{\varrho}(w) \frac{\partial}{\partial v}\left(\frac{2vw}{v + w}\right) \, dv \, dw \\
 & \quad  + \int \! f \left(\frac{2vw}{v + w}\right) \varrho(v) \tilde{\varrho}(w) \frac{\partial}{\partial w}\left(\frac{2vw}{v + w}\right) \, dv \, dw \\
&= \int \! f(x) \varrho(v(x)) \tilde{\varrho}(w) dx \, dw + \int \! f(x) \varrho(v) \tilde{\varrho}(w(x)) dx \, dv\\
&\le \left(\|\varrho\|_\infty  + \|\tilde{\varrho}\|_\infty \right) \|f\|_1.
\end{align*}
Hence $\|T_0(\varrho, \tilde{\varrho})\|_\infty \le \|\varrho\|_\infty  + \|\tilde{\varrho}\|_\infty$ and the lemma follows from the translation invariance of the norm.
\end{proof}

We will now consider the validity of Assumption~\ref{nonconcentration} for the special cases
\begin{itemize}
\item $c > 1$,
\item $V$ has a Gaussian distribution and $c > 1/2$,
\item $V$ has a Cauchy component and $c > 0$,
\end{itemize}
as mentioned in the introduction. The case $c > 1$ is an easy consequence of Lemma~\ref{easybound} since
\[ \|T_{p_r}...T_{p_1}\varrho_E\|_\infty \le 2^r \|\varrho_E\|_\infty = 2^{ r} \|\varrho\|_\infty\]
so Assumption~\ref{nonconcentration} is true with $I = \rr$ and $\delta = c-1 > 0 $.

Our analysis of the Gaussian distribution $\mathcal{N}(\mu, \sigma)$ is based on the following observations:
\begin{itemize}
\item If $V \in \rr^L$ is a random vector with independent $\mathcal{N}(0, \sigma)$ entries and $O: \rr^L \to \rr^L$ is an orthogonal matrix, then $OV$ also consists of independent $\mathcal{N}(0, \sigma)$ entries.
\item If $F: \cp \to \cp$ is a singular Herglotz function and $A \subset \rr$ is a Borel set, then
\[\int \! 1_A(F(t + i0)) P_z(t) \, dt = \int \! 1_A(t) P_{F(z)}(t) \, dt\]
where $P_z$ is the Poisson kernel corresponding to $z \in \cp$ (cf.~\cite{MR3405613}).
\end{itemize}
Let $\varphi_r = 2^{-r/2}(1, ..., 1) \in \rr^{|B_r|}$ be the unit vector with constant entries. By rotation invariance, there exists a random vector $Z \in \varphi_r^\perp$ and an independent scalar Gaussian $g \sim \mathcal{N}(0, \sigma)$ such that
\[V = g \varphi_r + \mu_r \varphi_r + Z\]
where $\mu_r = 2^{r/2}\mu$. Since there exist some $z \in \cp$ and $C < \infty$ such that the $\mathcal{N}(0,\sigma)$ density is dominated pointwise by $C \, P_z$, this implies that for any bounded Borel set $A \subset \rr$
\[\ee 1_A(\rg^r V) \le C\int \! 1_A(\rg^r (t \varphi_r +\mu_r \varphi_r + Z)) \,  P_z(t)  \, dt \, \xi(dZ),\]
where $\xi$ is some probability distribution on $\varphi_r^\perp$. Notice that $\rg^r (V)$ is a singular Herglotz function of each of the variables $V_0, ..., V_{2^r-1}$ with the property
\[\im \rg^r (V) \geq \min \{ \im V_k \, | \, 0 \le k \le 2^r - 1\}\]
which follows from the definition of $\rg (V)$ and the fact that
\[\im \left(\frac{1}{2z} + \frac{1}{2w} \right)^{-1} \geq \min \{ \im z,  \im w\}.\]
Thus $F(t) = \rg^r (t \varphi_r  + \mu_r \varphi_r + Z)$ is a singular Herglotz function of $t$  when $\mu$ and $Z$ are fixed. Hence
\begin{align*}\ee 1_A(\rg^r V) &\le  C \iint \, 1_A(t) P_{F(z)}(t) \, dt \,  \xi(dZ)\\
&\le \frac{C}{\im F(z)} |A|\\
&\le C 2^{r/2} |A|,
\end{align*}
which proves that $\|T_{p_r} ... T_{p_0}\varrho_E\|_\infty \le C \, 2^{r/2}$ uniformly in $E \in \rr$ because the previous estimates did not depend on $\mu$. Thus Assumption~\ref{nonconcentration} is true with $I = \rr$ and $\delta = c - 1/2$.

Finally, we consider the case where $\varrho$ is a mixture of Poisson kernels, i.e.,
\beq{cauchymixture} \varrho = \int_\cp \! P_z \, \mu(dz)
\eeq
for some probability measure $\mu \in M(\cp)$. A simple calculation, which is described in some detail in~\cite{MR2352276}, shows that
\[T_p\left(\int_{\cp} \! P_z \, \mu(dz) \right) = \int_{\cp} \! P_z \, T_p \mu (dz)\]
and that $\supp T_p\mu \subset \{z \in \cp \, | \, \im z > \epsilon\}$ if $\supp \mu \subset \{z \in \cp \, | \, \im z > \epsilon\}$. In particular, if $\varrho$ is of the form~\eqref{cauchymixture} with $\supp \mu \subset \{z \in \cp \, | \, \im z > \epsilon\}$, then
\[\|T_{p_r}...T_{p_1} \varrho_E\|_\infty \le \epsilon^{-1},\]
which proves Assumption~\ref{nonconcentration} with $I = \rr$ and $\delta = c$. By definition, $V$ has a Cauchy component if $\varrho = \mu \ast P_z$ for some $z \in \cp$ and some probability measure $\mu \in M(\rr)$, which is a special case of~\eqref{cauchymixture}.

\section{Eigenfunction Correlators and IPRs}\label{IPR}
The purpose of this appendix is to prove two statements made in the introduction regarding the behavior of the IPR in a regime of eigenfunction correlator localization. The arguments here do not rely on the specifics of the hierarchical model. First, let us present the proof of Corollary~\ref{cor:IPR}, which bounds the probability of the event
\[A = \left\{\begin{array}{c} \mbox{There is $ \psi \in \ell^2(B_n) $ with $ H_n \psi = \lambda \psi $ and} \\ \mbox{$\; | \lambda - E | \leq 2^{- n- 1} W $ such that $\;  P_2( \psi) \leq \varepsilon^4 $} \end{array} \quad \right\}.\]

\begin{proof}[Proof of Corollary~\ref{cor:IPR}]
Let $J_n = \{| \lambda - E | \leq 2^{- n- 1} W\}$ and let $\psi_\lambda$ denote the eigenfunction associated to an eigenvalue $\lambda \in \sigma(H_n)$. Then it follows from
\beq{iprhoelder}1 = \|\psi_\lambda\|_2^2 \le \|\psi_\lambda\|_4 \|\psi_\lambda\|_1
\eeq
that
\begin{align} \label{iprchebyshev}
\pp(A) &\le \pp\left (\sum_{\lambda \in \sigma(H_n) \cap J_n} \frac{1}{P_2(\psi_\lambda)^{1/4}} > \frac{1}{\varepsilon}\right )\\ \nonumber
& \le \varepsilon \, \ee \sum_{\lambda \in \sigma(H_n) \cap J_n} \frac{1}{\|\psi_\lambda\|_4} \le \varepsilon \, \ee \sum_{\lambda \in \sigma(H_n) \cap J_n}  \|\psi_\lambda\|_1 .
\end{align}
Since $\sum_{k \in B_n} |\psi_\lambda(k)|^2 = 1$, we have
\begin{align}\label{l1bound}
\ee 2^{-n} \sum_{\lambda \in \sigma(H_n) \cap J_n} \|\psi_\lambda\|_1 &= 2^{-n} \sum_{k \in B_n} \ee \sum_{j \in B_n} \sum_{\lambda \in \sigma(H_n) \cap J_n} |\psi_\lambda(k)|^2 |\psi_\lambda(j)|\\ \nonumber
&\le 2^{-n} \sum_{k \in B_n} \sum_{j \in B_n} Q_n(k,j; J_n)\\ \nonumber
&\le 2^{-n} \sum_{k \in B_n} C|J_n| = C|J_n|
\end{align}
in regimes of eigenfunction correlator localization. Plugging~\eqref{l1bound} into~\eqref{iprchebyshev}, we obtain
\begin{align*}
\pp(A) &\le \varepsilon \, \ee \sum_{\lambda \in \sigma(H_n) \cap J_n}  \|\psi_\lambda\|_1 \le \epsilon \, 2^n  C |J_n| = C \, W \varepsilon. 
\end{align*}
\end{proof}

Our other goal is to prove that eigenfunction correlator localization implies the lower bound~\eqref{eq:avgiprbound} for the averaged IPR
\[\Pi_n(I) = \frac{\ee \sum_{\lambda \in \sigma(H_n) \cap I} \|\psi_\lambda\|_4^4}{\ee \sum_{\lambda \in \sigma(H_n) \cap I} 1}.\]
Using~\eqref{iprhoelder} term by term yields
\[\Pi_n(I)  \geq \frac{\ee \sum_{\lambda \in \sigma(H_n) \cap I} \|\psi_\lambda\|_1^{-4}}{\ee \sum_{\lambda \in \sigma(H_n) \cap I} 1}.\]
We now apply Jensen's inequality with the probability measure defined by
\[\mu(f) = \frac{\ee \sum_{\lambda \in \sigma(H_n) \cap I} f(\lambda)}{\ee \sum_{\lambda \in \sigma(H_n) \cap I} 1}\]
and the convex function $\Phi(x) = x^{-4}$ to see that
\begin{align*}\Pi_n(I) &\geq \left( \frac{\ee \sum_{\lambda \in \sigma(H_n) \cap I} \|\psi_\lambda\|_1}{\ee \sum_{\lambda \in \sigma(H_n) \cap I} 1} \right)^{-4}\\
&=  \left( \frac{2^{-n}\ee \sum_{\lambda \in \sigma(H_n) \cap I} 1}{2^{-n}\ee \sum_{\lambda \in \sigma(H_n) \cap I} \|\psi_\lambda\|_1} \right)^{4}.
\end{align*}
The numerator of this expression is equal to
\begin{align*}
2^{-n} \ee \tr 1_I(H_n) &= 2^{-n} \sum_{k \in B_n} \ee \langle \delta_k, 1_I(H_n) \delta_k \rangle\\
& = \ee \langle \delta_0, 1_I(H_n) \delta_0 \rangle = \nu_n(I)
\end{align*}
so repeating the calculation~\eqref{l1bound} with $I$ in place of $J_n$ shows that
\begin{equation}\label{eq:IPRproof}
\Pi_n(I) \geq C^{-4} \left( \frac{\nu_n(I)}{|I|} \right)^4
\end{equation}
as desired.

\section{Spectral Localization}\label{Kritchevski}
This appendix contains the completion of an argument by E.~Kritchevski~\cite{MR2352276}, which proves that the spectrum of $H$ is almost surely of pure-point type with eigenfunctions satisfying
\[\sum_{k \in \nn} 2^{\frac{c}{4} d(0,k)} |\psi(k)|^2 < \infty\]
for any parameters $(\p, \varrho)$ without relying on Assumption~\ref{nonconcentration}. The following should be regarded as an accompanying note to~\cite{MR2352276} and thus we will not present the entire argument in detail, but simply cite the theorems of~\cite{MR2352276} as necessary. The argument makes use of the truncations $H_{n,m}$ in the $n \to \infty$ limit:
\[H_{\infty, n} = \sum_{r = 1}^n p_r E_r + V.\]
Notice that, for fixed realizations of $V$,
\[\|H - H_{\infty, n}\| \le \sum_{r = n+1}^\infty |p_r|\]
and hence
\[\lim_{n \to \infty} (H_{\infty, n} - z)^{-1} \delta_j = (H-z)^{-1} \delta_j\]
for any $j \in \nn$ and $z \in \cp$. We will be particularly interested in the quantities 
\begin{align*}G_n(j, k; z) &= \langle \delta_k,  (H_{\infty, n} - z)^{-1}  \delta_j \rangle\\
g_n(j; z) &= 2^{-n} \sum_{k \in B_n(j)} \langle \delta_k,  (H_{\infty, n} - z)^{-1}  \delta_j \rangle\\
Q_n(j, z) &= \langle \varphi_n(j),(H_{\infty, n} - z)^{-1}  \varphi_n(j) \rangle
\end{align*}
with $k \in B_n(j)$ and $\varphi_n(j) = 2^{-n/2}1_{B_n(j)}$. Proposition 2.2 of~\cite{MR2352276} contains the formula
\[G_n(j, k; z) = G_0(j, k; z) - \sum_{r = d(j, k)}^n 2^{r-1} p_r g_{r-1}(j; z)g_r(k; z).\]
Letting $w(k) = 2^{\mu d(j,k)} $ and using the triangle inequality for the $w$-weighted $\ell^2$-norm, this implies that
\[S(j, n, \mu) \vcentcolon = \left(\sum_{k \in \nn} w(k) |G_n(j, k; E)|^2 \right)^{1/2} \]
is bounded by
\begin{align*}
 |G_0(j,j; E)| + \sum_{r = 1}^n 2^{r-1} |p_r| |g_{r-1}(j; E)| \left(\sum_{d(j,k) \le r} \! w(k) |g_r(k; E)|^2 \right)^{1/2}\\
\le  |G_0(j,j; E)| + \sum_{r = 1}^n |p_r| 2^{\mu r} 2^{r-1}  |g_{r-1}(j; E)| \left(\sum_{d(j,k) \le r} \! |g_r(k; E)|^2 \right)^{1/2}.\\
\end{align*}
Provided that
\beq{kritchevskicondition}\sum_{r =1}^\infty |p_r Q_r(j, E)| < \infty
\eeq
for $\pp \otimes m$-almost every $(\omega, E) \in \Omega \times \rr$, Proposition 2.3 and the proof of Proposition 2.4 in~\cite{MR2352276} show that for almost every $(\omega, E)$ there exist constants $C(\omega, E), C^\prime(\omega,E) < \infty$ such that
\[ \left(\sum_{d(j,k) \le r}  |g_r(k; E)|^2 \right)^{1/2} \le C(\omega, E) \, 2^{\frac{c}{4} r}\]
and
\[2^{r-1}  |g_{r-1}(j; E)|  \le C^\prime(\omega, E)\]
for all $r \geq 1$. It follows that
\beq{greensfunctionbound} 
S(j, n, \mu)  \le |G_0(j,j; E)|+ C(\omega,E)C^\prime(\omega, E)\sum_{r = 1}^n 2^{-cr} 2^{\mu r} 2^{\frac{c}{4}r}.
\eeq
and since $G_0(j,j; E)$ exists and is finite for almost every $(\omega, E)$, choosing $\mu = 0$ we obtain
\[\sup_{n \geq 1} \|(H_{\infty, n} - E)^{-1} \delta_j\| = \sup_{n \geq 1} S(j,n, 0) < \infty\]
for almost every $(\omega, E)$. Applying the monotone convergence theorem to the spectral measures of $\delta_j$ for $H$ and $H_{\infty, n}$ shows that
\begin{align*}
 \lim_{\epsilon \to 0} \|(H - E - i\epsilon)^{-1} \delta_j\| &= \sup_{\epsilon > 0}\|(H - E - i\epsilon)^{-1} \delta_j\|\\
 &\le \sup_{\epsilon > 0} \sup_{n \geq 1} \|H_{\infty, n} - E - i\epsilon )^{-1} \delta_j\|\\
 &=\sup_{n \geq 1} \sup_{\epsilon > 0} \|(H_{\infty, n} - E - i\epsilon)^{-1} \delta_j\|\\
 &= \sup_{n \geq 1} \|(H_{\infty, n} - E)^{-1} \delta_j\ < \infty,
\end{align*}
and thus the Simon-Wolff Criterion~\cite{MR820340} asserts that the spectrum of $H$ is almost surely of pure-point type. If $G$ denotes the Green function of the full operator $H$, then
\begin{align*}
 |G(j,k; z) - G_n(j, k; z)| &= \left| \langle \delta_j, (H_{\infty, n} - z)^{-1} (H - H_{\infty, n}) (H-z)^{-1} \delta_k \rangle \right|\\
 &\le \|H - H_{\infty, n}\| \|(H_{\infty, n} - z)^{-1} \delta_j\| \|(H - z)^{-1} \delta_k\|
\end{align*}
and the preceding argument proves that for almost all $(\omega, E)$ we can take first $\im z \to 0$ and then $n \to \infty$ so that $G_n(j,k; E) \to G(j,k; E)$ for almost all $(\omega, E)$. Applying Fatou's lemma to~\eqref{greensfunctionbound} with $\mu = \frac{c}{4}$ we see that
\begin{align*}
 \left(\sum_{k \neq 0} 2^{\frac{c}{4} d(j,k)} |G(j, k; E)|^2 \right)^{1/2} & \le C(\omega,E)C^\prime(\omega, E) \sup_{n \geq 0} \sum_{r = 1}^n 2^{-\frac{c}{2}r}\\
 &< \infty,
\end{align*}
and a trivial modification of Theorem 9 from ~\cite{MR820340} now shows that the eigenfunctions of $H$ satisfy 
\[\sum_{k \in \nn} 2^{\frac{c}{4} d(0,k)} |\psi(k)|^2 < \infty\]
as well.

Thus it remains to prove~\eqref{kritchevskicondition}. Since $Q_r(j; z)$ is the Borel transform of a singular probability measure, Boole's inequality (\cite[Prop. 8.2]{MR3364516}) shows that
\[(\pp \otimes m)(\{|Q_n(j; E)| > 2^{\frac{c}{2} r}\}) = 2 \cdot 2^{-\frac{c}{2} r}.\]
It follows from the Borel-Cantelli lemma that
\[|Q_n(j; E)| \le C^{\prime \prime}(\omega, E) 2^{\frac{c}{2} r} \]
with $C^{\prime \prime}(\omega, E) < \infty$ on a set of full $\pp \otimes m$ measure, and this implies~\eqref{kritchevskicondition}.

\subsection*{Acknowledgment}
This work was supported by the DFG (WA 1699/2-1).
\bibliographystyle{abbrv}
\nocite{*}
\bibliography{References}
\end{document}